\documentclass[12pt]{article}
\usepackage{amsfonts, amssymb, amsmath, amsthm}
\usepackage[onehalfspacing]{setspace}
\usepackage[round, longnamesfirst]{natbib}
 \usepackage{tikz}
  \usepackage{epigraph}
 \usepackage{tikzpeople}
\usepackage{multirow}
\usepackage{comment}
\usepackage[english]{babel} 
\usepackage{amsmath,amsfonts,amsthm,bm} 
\usepackage{array}
\newcolumntype{?}{!{\vrule width 2pt}}
\newlength{\Oldarrayrulewidth}

\usepackage[justification=centering,textfont={sc},labelfont={sc}]{caption,subcaption}
\newcommand{\interior}[1]{%
  {\kern0pt#1}^{\mathrm{o}}%
}

\usepackage[capposition=bottom]{floatrow}
\newtheorem{theorem}{Theorem}

\newtheorem{definition}{Definition}
\newtheorem{lemma}{Lemma}
\newtheorem{proposition}{Proposition}

\usepackage{apptools}
\AtAppendix{
  \counterwithin{lemma}{section}
  \counterwithin{remark}{section}
  \counterwithin{corollary}{section}
  \counterwithin{definition}{section}
  \counterwithin{proposition}{section}
}

\usepackage{pgfplots}
\usepgfplotslibrary{fillbetween}
\usetikzlibrary{intersections}
\pgfplotsset{compat=1.13}
\makeatletter
\renewenvironment{proof}[1][\proofname]{\par
  \pushQED{$\blacksquare$}%
  \normalfont \topsep6\p@\@plus6\p@\relax
  \trivlist
  \item[\hskip\labelsep
        \bfseries
    #1\@addpunct{:}]\ignorespaces
}{%
\hfill \popQED\endtrivlist\@endpefalse
}
\makeatother

\renewcommand{\interior}[1]{\text{int}({#1})}
\usepackage[hidelinks, colorlinks = true,
            linkcolor = black,urlcolor  = blue,
            citecolor = black,
            anchorcolor = blue]{hyperref}
\usepackage{cleveref}
\crefname{subsection}{subsection}{subsections}
\Crefname{subsection}{Subsection}{Subsections}
\usepackage{enumerate}
\usepackage{appendix}

\begin{document}

\begin{titlepage}
\title{Sorting and Grading\footnote{We thank the editor, Nicola Pavoni, and four anonymous referees for comments and suggestions that greatly helped us improve this paper. We   
 are very grateful to
Arjada Bardhi, Antonio Cabrales, 
Matteo Escud\'{e}, 
B\r{a}rd Harstad, Nenad Kos,  
Eirik  Kristiansen,   Ming Li, Mart\'{i} Mestieri, Eduardo Perez-Richet, Johannes Schneider, Vasiliki Skreta, Takuro Yamashita, and Andy Zapechelnyuk  for   useful  discussions and feedback.}}
\author{Jacopo Bizzotto\thanks{Oslo Business School, OsloMet. Email: \href{mailto:jacopo.bizzotto@oslomet.no}{jacopo.bizzotto@oslomet.no}.} \quad Adrien Vigier\thanks{School of Economics, University of Nottingham. Email: \href{adrien.vigier@nottingham.ac.uk}{adrien.vigier@nottingham.ac.uk.}}}
\date{February 2024}
\maketitle
\begin{abstract}

We 
 propose a  framework  to assess how to optimally sort and grade 
 students of heterogenous ability.
Potential  employers  face  uncertainty regarding an  
individual's productive value. Knowing which school an individual went to is useful for two reasons: firstly, average student ability may  differ across schools; secondly, different schools  may use different grading rules and thus provide varying incentives to exert effort.  An optimal  school system exhibits coarse stratification with respect to ability,  and more lenient grading at the top-tier schools
than at the bottom-tier schools. Our paper  contributes to the ongoing  policy debate on  tracking in secondary schools.
\end{abstract}

\vspace{.2 in}



\thispagestyle{empty}
\end{titlepage}

\newpage

\section{Introduction}


Few debates are more heated than  how best to 
 foster students' achievement.\footnote{While the discussion in this paper focuses on secondary school systems, another application of our  analysis   concerns  public universities.}
Two questions occupy a central place in this debate:  (i) should students be sorted based on ability?  (ii) how should students be graded? 
In this paper, we 
study the interplay between sorting and grading, and argue that one reason for sorting students is  that, depending on their ability,   students ought to be graded  differently.

The practice of separating students according to ability, commonly  referred to as tracking,  typically  occurs in high school.\footnote{The age at which  students are first tracked varies across countries. See, e.g.,  \cite{Betts2011Review} or \citet{Woessmann2009Evidence}.}
In England, for example,  the public secondary school system includes a limited number of selective grammar schools together with a vast number of comprehensive (i.e. non-selective) schools.\footnote{For an history of tracking in the United Kingdom, see \citet{WhittyPower2015}.} 
Countries like Austria, Germany and Switzerland divide  students into 
  academic and vocational high schools.
By contrast,  in  the United States and  Canada,
 tracking takes the form of ability grouping within   schools.   Yet
other countries  choose not to track students at all, including France and  Italy. The proponents of tracking 
argue that grouping  students by  ability is socially efficient and 
benefits 
 all  categories  of students;  its  critics claim on the other hand   that only high-ability students benefit from
tracking and that,   
if selection supposedly based on ability  ends up capturing students' socio-economic background, tracking might be a factor which contributes to  perpetuate inequality across generations.  
Empirical evidence on the matter  is mixed: \cite{galindo2007heterogeneous} and
   \citet{Duflo2011Kenya} report  evidence of the positive effects  of tracking; 
others, including 
  \cite{Kerr2013Cognitive} and 
   \cite{matthewes2021better}, 
document a negative impact overall.

The issue of grading 
is no less controversial than that of tracking. While the importance of grading standards for student behavior is generally accepted,  disagreement is common regarding   what these standards ought to  be. 
One  recurring theme in this  debate concerns  the merits and flaws  of standardized tests.\footnote{Standardized tests are  defined as those in which the conditions and content are equal for everyone taking the test, regardless of when, where, or by whom the test was given or graded.} The  proponents   of such  tests emphasize that they are fair and    facilitate 
 cross-student comparisons. 
Critics argue that different students ought to be graded differently. For example, 
\cite{becker1992learning}
argue that stratifying students into groups in which each has a chance for
success may be preferred to the setting of a single national standard.

We 
 propose a simple framework  to assess how to optimally sort and grade 
students.
A population  of students is divided across  the   different schools forming part of a given school system.
Each individual comes out of school with a  productive value that can be either  high or low. The 
distribution of an individual's productive value is determined by  a combination of his ability and effort. However, 
ability, effort,  and productive value are all unobserved
by employers.
Employers only observe the school which the student attended and the grade
which he obtained; they form beliefs about the  student's productive value and   offer a competitive wage.\footnote{In practice,  some individuals  seek employment immediately after secondary school while others go  to university. What we generically refer to as employers could just as well represent  university admissions officers.}
Figure \ref{Figure1ok} summarizes  the model.  

\begin{figure}
\begin{tikzpicture}[x=3.5cm,y=6.5cm,scale=1, house/.style={
        draw, minimum width=12mm, minimum height=10mm, 
        label={[name=labelaux, inner sep=1mm, distance=0pt]#1}, 
        append after command={%
            \pgfextra \draw ([xshift=.5\pgflinewidth]\tikzlastnode.north west)--(labelaux.north)
                --([xshift=-.5\pgflinewidth]\tikzlastnode.north east);\endpgfextra
        }
    },
    vertex/.style={draw,circle,minimum size=0.8cm}
    ,
    grade/.style={draw,rectangle, minimum width=8mm, minimum height=12mm,}]
  \draw(0,1) node{STUDENTS};
  \draw(1,1) node{SCHOOLS};
  \draw(2,1) node{\begin{tabular}{c}
    PRODUCTIVE\\
    VALUES 
\end{tabular}};
  \draw(3,1) node{GRADES};
  \draw(4,1) node{\begin{tabular}{c}
    EMPLOYERS'\\
    BELIEFS 
\end{tabular}};

\node[person,shirt=black,scale=1](A)  at (0,.9)  
{};
\node[person,shirt=blue,scale=1](B)  at (0,.6)  
{};
\node[person,shirt=red,scale=1](C)  at (0,.3)  
{};
\node[person,shirt=yellow,scale=1](D)  at (0,0) 
{};
\node[house](S2) at (1,.1) {};
\node[house](S1) at (1,.7) {};

\node[vertex](pv1) at(2,.8){1};
\node[vertex](pv2) at(2,.6){0};
\node[vertex](pv3) at(2,.2){1};
\node[vertex](pv4) at(2,.0){0};

\node[grade](g1) at(3,.8){pass};
\node[grade](g2) at(3,.6){\hspace{.8mm}fail\hspace{.8mm}};
\node[grade](g3) at(3,.2){pass};
\node[grade](g4) at(3,.0){\hspace{.8mm}fail\hspace{.8mm}};

\draw[|-|,color=black] (4,0) -- (4,.1);
\draw[-|,color=black] (4,0.1) -- (4,0.3);
\draw[-|,color=black] (4,0.3) -- (4,0.5);
\draw[-|,color=black] (4,0.5) -- (4,0.6);
\draw[-|,color=black] (4,0.6) -- (4,0.85);
\node at(3.95,0){0};
\node at(3.95,.85){1};

\node(B4) at(4,0.1){};
\node(B3) at(4,.3){};
\node(B2) at(4,.5){};
\node(B1) at(4,.6){};


\path[->, auto] (A) edge node {} (S1)
(B) edge node {} (S1)
(D) edge node {} (S2);
\path[->] (C) edge node {} (S2);
\path[->, dashed] (S1) edge node {} (pv1)
edge node {} (pv2)
 (S2) edge node {} (pv3)
 edge node {} (pv4);
\path[->, dashed] (pv1) edge node {} (g1)
  edge node {} (g2)
(pv2) edge node {} (g1)
  edge node {} (g2)
(pv3) edge node {} (g3)
  edge node {} (g4)
  (pv4) edge node {} (g3)
  edge node {} (g4);
  \path[->, dashed] (pv1) edge node { } (g1)
  edge node {} (g2)
(pv2) edge node { } (g1)
  edge node {} (g2)
(pv3) edge node {} (g3)
  edge node {} (g4)
  (pv4) edge node {} (g3)
  edge node {} (g4);
 \path[->, auto] (g1) edge node { } (B1)
(g2) edge node { } (B3)
(g3) edge node {} (B2)
  (g4) edge node {} (B4);
  
\end{tikzpicture}\caption{The Setting}\label{Figure1ok}
\end{figure}

The assignment of students to schools 
matters for two reasons:
(a) the school that a student went to conveys information 
regarding   his productive value; (b)  different schools may use different grading rules.
We study the  problem of a planner  designing  the   school system so as to  
maximize social welfare. Since   due to moral hazard each student's equilibrium 
  effort is below the socially optimal level,  the problem of the planner in essence amounts to maximizing total effort.

The basic  question we address  is the following: should students be sorted  according to ability and,  if so, how?  
The  trade-off is as follows:
\begin{itemize}
    \item 
On the one  hand, 
 sorting  students enables tailored incentives. This is valuable, as the optimal grading rule for any given school depends on the average ability of the students assigned to that school.  The basic idea 
 is simple.
 For a student who is regarded as ``bad", the value of obtaining a  good grade is very high 
 as long as this grade  
is reliable 
  (i.e., as long as the grading rule   generates few false positives).
 By contrast,  in the case of students who are regarded as ``good", the logic is reversed.
In consequence, a one-size-fits-all approach to grading might end up depressing
all students' incentives to exert effort in school.
\item 
On the other hand, two schools with  different average student ability and the same grading rule should be merged. The logic is as follows.
Mixing
students increases the uncertainty faced by employers, which,  in turn, augments the
weight of school grades in employers' inferences. Yet, 
 when grades matter more, top marks
become worth the  effort. Thus,  mixing students indirectly favors their effort.

\end{itemize}

At an optimum, 
the aforementioned trade-off is resolved as follows:
(i)  students whose ability is above a cutoff are separated from those whose ability is below;
  (ii) on each side of the separating  cutoff, 
    all students 
    are   pooled.
 Bottom-tier students are evaluated through a tough grading rule (i.e., one generating
few false positives), while top-tier students are evaluated through a lenient grading rule
(i.e., one generating few false negatives).

To convey the paper's main  insights in the simplest possible way and ensure  the tractability of the model,  
the framework  we propose makes several assumptions of note.
Firstly, we ignore the fact that some schools may be intrinsically  better than others, perhaps because they have  more resources, better teachers, or a combination of both.
Secondly, we ignore 
 ``direct" peer effects. 
 For example, if  students learn from their peers,    one's   educational attainment  
 should  be positively affected by an increase in  the mean student ability  of the  school to  which he was  assigned. Alternatively, students may  have homophilous preferences, and learn best when grouped with  students of similar ability as their own.\footnote{It is well   known
that, in the presence of
direct
   peer  effects, the
optimal allocation of students  depends in a
complex way on the curvature properties of the educational production
function. See, e.g.,  \cite{arnott1987peer}. }
 Thirdly, we assume that an individual's productive value is   binary.\footnote{An example with three productive values is presented  in Section  \ref{sec:morethantwo}. } 
  Enriching the model along any  of the aforementioned  dimensions   
  would evidently improve the realism of our framework.

Last but not least, 
one critical    premise of our model,  is that the  planner can   assign  students to schools  on the basis of their ability. Austria, England,   Germany  and Switzerland 
provide examples of this kind. However,  in countries like  
France and Italy, students are typically free to  choose their secondary  school.\footnote{In France and Italy, this choice is somewhat constrained by a student's place of residence.} 
More generally, 
 social planners may either  lack information regarding  students' ability,
or be unable to use that information (e.g., for  political reasons).  We therefore examine  a variant of the model in which  students can choose the school they go to. In this variant of the model,  a school system  needs to be incentive compatible for all students in order to be implementable.
We show first that without   transfers (e.g., differentiated school fees or scholarships) the planner can do no better than randomly allocating  students to schools.
We then show that, with transfers, any constrained-optimal school system has the same qualitative features as when there  are no  incentive compatibility constraints: students whose ability is above a cutoff are separated from those whose ability is below, and,  on each side of the separating  cutoff, 
    all students 
    are   pooled.
However, the separating cutoff is larger.
In other words, the incentive compatibility constraints reduce the size of the top tier while increasing  that of the bottom tier.

The rest of the paper is organized as follows.  We next discuss the related literature. The model is presented in 
Section
\ref{sectionmodel}. Section \ref{section:eqeffort}
 contains the preliminary analysis.
Our main result is presented in Section \ref{sec:mainresult}. Section 
\ref{sec:IC} studies incentive compatibility. 
Section
\ref{sec:morethantwo} examines 
an example with 
three productive values, while  Section 
\ref{sec:conc} concludes.

\paragraph{Related Literature.}

We bridge two strands of literature: the literature on tracking, and that on grading.
The (mostly empirical)    literature on tracking
goes at least as far back as 
\cite{cook1924study} and \cite{whipple1936grouping}.
The arguments  in favor of tracking are typically twofold (see, e.g., \cite{slavin1990achievement} or \cite{Betts2011Review}). Firstly,  more homogeneous groups of students allow teachers to tailor their pedagogical strategies.
Secondly, with tracking, a 
students' performance  is only compared to that of similar-ability peers, preventing a possible lowering of  self-esteem that could result from comparisons with the work of higher ability students,  while allowing  highly gifted students to  be appropriately challenged and to view their own abilities more realistically.  
We
contribute to the aforementioned literature by
arguing  that another natural  reason  for which tracking could  be efficient is that tracking allows incentives to  be tailored to students' characteristics, via targeted  grading policies.\footnote{A different  argument  in support of tracking is as follows.  Suppose firms are heterogenous and the efficient matching between firms and  students is  assortative.
Then tracking could help achieve the  efficient matching. 
We view all of these arguments as complementary.}

The literature on grading can be divided in two strands.
A first strand   studies the link between   grading standards and    effort. 
\citet{becker1992learning} and 
  \citet{Betts98} argue   
that a rise in 
grading standards  has ambiguous effects on overall  achievement. While more able students  increase effort
in order to reach the new standard,   less able students, who must increase effort by a greater amount to
increase their achievement, may give up instead. 
\cite{costrell1994simple}  concludes that   social 
planners with 
egalitarian concerns 
should thus  lower grading standards, rather  than raise them.   \cite{himmler2013double}  further argue that 
 schools with  disadvantaged students ought to  lower standards more than others so as to compensate  a lack of incentives to graduate.\footnote{\citet{DubeyGeanakoplos10} and \citet{Zubrickas15} take a different approach and  
show that the use of coarse grades is optimal 
 when  students
care primarily about  their relative rank in class.}
The second  strand of the literature on grading 
views grade not primarily as a means to incentivize effort, but as a strategic tool for  schools to maximize  their 
 students' payoffs on the  job market.  In this spirit,
 \cite{chan2007signaling}, \cite{ehlers2016honest}, and   \cite{tampieri2019students} 
 develop  signaling theories  of grade inflation;
 \cite{OstrovskySchwarz10} explore information disclosure in matching markets
 and argue  that  schools may suppress some information about students in order
to improve their average job placement; 
\cite{BoleslavskyCotton15} study  information disclosure in  contests and show that, in  equilibrium,  schools  adopt grading policies
that do not perfectly reveal ability.
  We contribute to the literature on grading by exploring  the interplay between grading design and the allocation of students to schools.

Our paper is intimately connected to 
  the  vast   literature studying  statistical discrimination in labor markets, understood as employers' reliance 
on  payoff-irrelevant characteristics 
for making inferences  about workers' productivity.\footnote{For a detailed review of the   literature on the subject, see e.g. \cite{onuchic22}. }
 In the Arrowian tradition (\cite{arrow1973}), we study an environment in which  productivity is endogenous.
 As in \cite{coate1993will}, 
 because in our model employers' prior beliefs  affect their interpretation of the  productivity signal (in our case, the grade), 
employers' conjectures affect  
 incentives to invest in  skill acquisition. 
  We depart from the literature
following \cite{arrow1973}
 in two crucial ways. Firstly,  statistical discrimination does not proceed from equilibrium   multiplicity in our model: each school induces  a unique equilibrium level of effort; said  effort, however,   varies  with the way students are graded. Secondly,
whereas the  aforementioned literature takes 
the payoff-irrelevant characteristic for making inferences  about a  worker's productivity 
as given 
(e.g., race or gender),
in our model said  characteristic (namely, the school which a student  attended) is either   a choice made by  the student,  or part of the decision maker's  design problem.   
 In this sense, our model  is closer to Spence signalling (\cite{spence1973press}). However,  the connection with signalling stops here:  our model involves both  adverse selection and moral hazard, and employers do not  care  exclusively  about students' types, but   about  a combination of type and effort (productive values).

Our work is  also  connected to   
the   literature  on mechanism design
 with peer effects, comprising
\citet{jehiel1996not,jehiel1999multidimensional}, \cite{jehiel2001efficient},
\cite{board2009monopolistic}, \citet{rayo2013monopolistic},
\citet{rothschild2013redistributive,rothschild2016optimal}, and \cite{yamashita2021large}. Relative to this literature, 
the main novelty of our work
is that how much an individual values being part of  a certain  group is not fixed, but instead depends on  a   grading rule which policy makers can adjust.\footnote{The problem we analyze can  be viewed as an information design problem with constraints. Effectively,
 the school system plays the role of an information  structure for employers.}
Finally, the problem we analyze in Section \ref{sec:IC} connects our work to 
a number of recent papers that share with ours the feature  that the planner/designer  offers  a menu of items comprising a statistical experiment of some kind. This literature includes
 \citet{esHo2007optimal},  \cite{kolotilin2017persuasion}, \cite{li2017discriminatory},  \citet{bergemann2018design},  \cite{guo2019interval},  \cite{wei2020reverse}, \cite{yamashita2021type},  and  \cite{yang2022selling}. 
In all those papers, 
the experiment  
helps   the agent  solve a decision problem.  Instead,
in our setting 
the experiment (i.e., the grading rule)   helps potential employers evaluate the productive value of the  agent (the student).


\section{Model}\label{sectionmodel}

The broad features of our model are as follows.
A 
 population of students is assigned to different   schools.  At the end of school, a   student's value to an employer  can  be either high or low.
A student's chances of becoming high value 
are increasing  in 
ability and effort.
However, neither   ability nor effort are observed by employers, who  form beliefs about a student's productive value based  solely on the school which this student   attended and the grade which he obtained. Every  student then receives a wage which reflects employers' perception about his productive value.  The 
 details of the model are presented below.

 There are $m \in \mathbb{N}$ students, labelled by $j \in \{1,\cdots, m\}$. The assignment of students to schools is
 captured by a pair $(n, \phi)$, where 
 $n \in \mathbb{N}$ 
and  
  $\phi$
 is a surjective mapping from  
  $\{1, \cdots, m\}$ to  $\{1, \cdots, n\}$. We refer to such  a pair $(n, \phi)$ as an \textit{assignment}; under this assignment,  student  $j \in \{1,\cdots, m\}$ is assigned to school  $\phi(j)$.

Every  student  is  characterized by his  \textit{type}, which can be thought of as capturing   his ability. The type of student $j$ is denoted by 
 $ \theta_j \in(0, \alpha)$, where $\alpha \in (0,1)$. 
At the end of school, 
each student's  \textit{productive value} 
is either 0 or 1.
 The probability 
 that
 student $j$'s productive value is $1$
  equals   the sum of  this student's type $\theta_j$ 
and \textit{effort} $e_j \in [0, 1-\alpha]$.
  The cost of exerting  effort $e_j$ is given by 
$c(e_j)$, where  $c'>0$, $c''>0$,  $c'''\geq 0$,
$c'(0)=0$, 
and  $c'(1-\alpha)>1$.

Letting $b$ denote a parameter in $(0,1)$, we refer to 
a pair $G=(g^0, g^1) \in [0,1]^2$ such that 
$ g^1> g^0$ and  $g^1 -g^0\leq b$  as a
\textit{grading rule}.
The interpretation is as follows.
Every 
school delivers a grade to each of its students,  either
  ``pass" or ``fail".
A school's grading rule determines the probabilities of obtaining the grade pass conditional  on one's productive value: under  the grading rule  $G$, a   student with productive value 0  passes with probability $g^0$, while  a student with productive value 1 passes with probability $g^1$.   
The condition
$ g^1> g^0$  ensures that obtaining the  grade pass conveys more favorable information about one's  productive value than obtaining the grade fail;
the
condition  $g^1 -g^0\leq b$ 
 rules out perfectly informative grading rules. We 
denote by  $\mathbb{G}$  the  set of  all  grading rules.

We refer to 
a  tuple $\big(n, \phi, 
(G_{s})_{s=1}^n \big)$ composed of an assignment  $(n, \phi)$ and a grading rule $G_s$ for every school  $s \in \{1, \cdots, n\}$  as  a \textit{school 
system}.
Given a school system $\big(n, \phi, 
(G_{s})_{s=1}^n \big)$,  
the number and average type of students attending school   $s  \in \{1, \cdots, n\}$ will be denoted by $m_{s}$ and $\Bar{\theta}_{s}$, respectively.
We assume that, for every  $s \in \{1, \cdots, n\}$,
both $\Bar{\theta}_{s}$ and $G_s$
   are   public information. However, a
student's  type, effort, and productive value 
are all unobserved by employers.
Employers
 only observe
 the school which a student attended, and 
 the grade which he obtained. They  form conjectures about effort, and ultimately pay each 
  student $j$  a competitive wage equal to their (Bayesian)  posterior  belief  that this student's productive value is  equal to 1.

\section{Equilibrium Effort}\label{section:eqeffort}

This section examines the effort exerted by different students in   an arbitrary school system.
We show first 
 that each school system induces a unique  effort profile. We then 
show that 
the grading rule 
maximizing effort in a school depends on  its average  student type. Finally,
we show that
if two  schools share the same grading rule then pooling their students in a single school   increases  total  effort.

We start by defining the notion of equilibrium effort.

\begin{definition}
Given a  school  system $\big(n, \phi, 
(G_{s})_{s=1}^n \big)$, 
  $e$
 is an  \emph{equilibrium effort}  in  school $s \in \{1, \cdots, n\}$ if  choosing 
effort $e$
maximizes  the expected payoffs of  the    students  attending school  $s$ given  employers' conjecture that students of $s$ exert effort $e$. 
\end{definition}

Consider a school    $s$
 forming part of a given school system. 
Let $q_s^+(\hat{e}_s)$ 
denote employers' posterior belief that a student of $s$ who obtained the grade pass has productive value 1, given 
employers' conjecture that students of $s$ exert effort $\hat{e}_s$.
Let   $q_s^-(\hat{e}_s)$ denote the corresponding  belief for   students of $s$ who obtained the grade fail.
Then
\[ q_s^+(\hat{e}_s)
= \frac{ (\Bar{\theta}_s + \hat{e}_s) g_s^1}{ (\Bar{\theta}_s + \hat{e}_s) g_s^1 +  (1-\Bar{\theta}_s - \hat{e}_s) g_s^0}, \quad 
q_s^-(\hat{e}_s)= \frac{ (\Bar{\theta}_s + \hat{e}_s) (1-g_s^1)}{ (\Bar{\theta}_s + \hat{e}_s)  (1-g_s^1) +  (1-\Bar{\theta}_s - \hat{e}_s)  (1-g_s^0)}.
\]

A type-$\theta$
student 
attending school $s$ and exerting effort $e $ obtains the grade 
pass with probability $(\theta+e)g_s^1+ (1-\theta-e) g_s^0$, and the grade fail otherwise.
So letting $U_s(\theta, e, \hat{e}_s)$
denote the expected payoff of a
 type-$\theta$
student 
attending school $s$ and exerting effort $e$ 
(given 
employers' conjecture that students of $s$ exert effort $\hat{e}_s$):
\[ U_s(\theta, e, \hat{e}_s)=
[(\theta+e)g_s^1+ (1-\theta-e) g_s^0] q_s^+(\hat{e}_s) + 
[1-(\theta+e)g_s^1- (1-\theta-e) g_s^0]  q_s^-(\hat{e}_s)  -
 c(e).\]
The unique   $e$
maximizing 
$ U_s(\theta, e, \hat{e}_s)$
satisfies 
\[ c'(e)=
(g_s^1 -  g_s^0) \big(q_s^+(\hat{e}_s)  - q_s^-(\hat{e}_s)\big) = B(G_s, \Bar{\theta}_s+ \hat{e}_s ), 
\]
where 
\[
B(G,x):= (g^1 -  g^0) \bigg( \frac{x g^1  }{ g^0 + x (g^1 -  g^0) }   - \frac{x(1-g^1)}{ 1- g^0  -x (g^1 -  g^0) } \bigg), \quad  \forall x\in (0,1).
\]
 It follows that 
$e$ is an equilibrium effort in school    $s$ if and only if
\begin{equation}\label{eq:focbase}
    c'(e)=B(G_s, \Bar{\theta}_s+ e ).
\end{equation}
The function $e \mapsto B(G_s, \Bar{\theta}_s+ e )- c'(e)$ is 
 strictly concave, positive  at $e=0$, and negative  at $e=1-\alpha$. So equation \eqref{eq:focbase} has a unique solution,
 and this solution belongs to  
   $(0,1- \alpha)$; as it  
 is determined by $\bar{\theta}_s$ and $G_s$, we will denote it by $\xi(\Bar{\theta}_s, G_s)$.  The next lemma summarizes the main takeaway.

\begin{lemma} To  each   school
 $s$ 
forming part of a given school system 
is associated  a unique equilibrium effort  $\xi(\Bar{\theta}_s, G_s)$. This equilibrium effort is the unique solution of equation  \eqref{eq:focbase}.
\end{lemma}

Henceforth, 
let
$G_{\star}$
denote the grading rule
$(1-b,1)$, and $G^{\star}$ the grading rule $(0, b)$.
The grading rule 
 $G_{\star}$  thus passes every  student with productive value $1$; 
 the grading rule  
$G^{\star}$, on the other hand, fails  every  student 
with productive value $0$.
Moreover, note that   $G_{\star}$ and $G^{\star}$ both belong to the set of 
  maximally ``informative" grading rules within $\mathbb{G}$.\footnote{Consider two grading rules,  $G_A$ and $G_B$, such that either $g^1_A= g^1_B$ or  $g^0_A= g^0_B$. Then $G_A$ is Blackwell-more-informative than $G_B$ if and only if  
$g^1_A- g^0_A> g^1_B- g^0_B$. }
We can now state our first proposition.

\begin{proposition}\label{prop1}
For any   $\Bar{\theta} \in (0, \alpha)$ and any  grading rule  $G \in \mathbb{G} \setminus \{G_{\star}, G^{\star}\}$, 
either $\xi(\Bar{\theta}, G_{\star})>
\xi(\Bar{\theta}, G)$ or  $\xi(\Bar{\theta}, G^{\star})>
\xi(\Bar{\theta}, G)$.
\end{proposition}

\begin{proof} Let  $\Bar{\theta} \in (0, \alpha)$. 
For any $x \in (0,1)$,
the function   $g^1 \mapsto  B\big( (g^0, g^1),x \big)$
is increasing, and 
the function   $g^0 \mapsto  B\big( (g^0, g^1),x \big)$
decreasing. So either $g^1- g^0=b$ or we can find $\Tilde{G} \in \mathbb{G}$ such that  
$B(\Tilde{G}, x)> B(G, x)$.
In particular, either 
 $g^1- g^0=b$ or we can find $\Tilde{G} \in \mathbb{G}$ such that  
\begin{equation}\label{eqcupcaf2}
B\big(\Tilde{G}, \Bar{\theta} +\xi(\Bar{\theta}, G) \big)>  B\big(G, \Bar{\theta} +\xi(\Bar{\theta}, G) \big)=c'\big( \xi(\Bar{\theta}, G)\big).
\end{equation}
The function $e \mapsto B(\Tilde{G}, \Bar{\theta}+ e )- c'(e)$ is 
 strictly concave, positive  at $e=0$, and zero 
at   $e= \xi(\Bar{\theta}, \Tilde{G})$. So  
\eqref{eqcupcaf2} implies $ \xi(\Bar{\theta}, G)<\xi(\Bar{\theta}, \Tilde{G})$.
In sum, either 
 $g^1- g^0=b$ or we can find $\Tilde{G} \in \mathbb{G}$ such that  
 $ \xi(\Bar{\theta}, \Tilde{G}) > \xi(\Bar{\theta}, G)$.

Next, 
for any $x \in (0,1)$,
the 
function   $g^0 \mapsto  B\big( (g^0, g^0+b),x \big)$
is strictly convex. So if 
 $G= (g^0, g^0+b)\in \mathbb{G} \setminus \{G_{\star}, G^{\star}\}$, we can find $\Tilde{G} \in  \{G_{\star}, G^{\star}\}$ such that 
$B(\Tilde{G}, x)> B(G, x)$.  In particular, 
 if 
 $G= (g^0, g^0+b)\in \mathbb{G} \setminus \{G_{\star}, G^{\star}\}$, we can find $\Tilde{G} \in  \{G_{\star}, G^{\star}\}$ such that \eqref{eqcupcaf2}  holds.
Then reasoning as we did above yields
 $ \xi(\Bar{\theta}, \Tilde{G}) > \xi(\Bar{\theta}, G)$.
\end{proof}

By 
Proposition \ref{prop1}, if a grading 
rule 
$G$ maximizes  effort in a school,  either $G=G_{\star}$ or $G= G^{\star}$.  The logic is as follows.
Any   grading rule $G \notin \{G^{\star}, G_{\star}\}$
can be written as  a strict  convex combination of $G^{\star}$, $G_{\star}$,  and some  uninformative grading rule,  $G^{\varnothing}$ say.\footnote{We here slightly abuse the terminology, since, the way we have defined them, every grading rule is informative. } 
So using the grading rule $G$ is equivalent  to randomizing between  
$G^{\star}$, $G_{\star}$, and  $G^{\varnothing}$. Yet, randomizing the grading rule is conducive to moral hazard. So there exists a grading rule in  the set 
$\{G^{\star}, G_{\star}, G^{\varnothing}\}$ such that adopting this grading rule 
for sure  induces more effort than $G$.\footnote{And,  $G^{\varnothing}$  being uninformative,  either 
$G^{\star}$ or $G_{\star}$  must induce more effort than $G$.}

The next proposition tells us when 
$G^{\star}$ induces more effort than $G_{\star}$, 
 or the other way around.

\begin{proposition}\label{prop2}
    There exists $\Bar{\theta}^\dagger$ such that  $\xi(\Bar{\theta}, G^{\star})> \xi(\Bar{\theta}, G_{\star})$ for $\Bar{\theta}<\Bar{\theta}^\dagger$,  and 
$\xi(\Bar{\theta}, G_{\star})> \xi(\Bar{\theta}, G^{\star})$ for $\Bar{\theta}>\Bar{\theta}^\dagger$.
\end{proposition}

\begin{proof}
The function $x \mapsto B( G^{\star},x)$ is decreasing. 
Hence, for   $0<\Bar{\theta}_A < \Bar{\theta}_B< \alpha$, we have 
\begin{equation}\label{eqcupcaf3}
B \big(G^{\star}, \Bar{\theta}_A+ \xi(\Bar{\theta}_B, G^{\star})  \big) >
 B \big(G^{\star}, \Bar{\theta}_B + \xi(\Bar{\theta}_B, G^{\star}) \big)=c' \big(  \xi(\Bar{\theta}_B, G^{\star}) \big).
\end{equation}
The function  $e \mapsto B(G^{\star}, \Bar{\theta}_A+ e )- c'(e)$ is 
 strictly concave, positive  at $e=0$, and zero 
at   $e= \xi(\Bar{\theta}_A, G^{\star})$. So  
\eqref{eqcupcaf3} implies $ \xi(\Bar{\theta}_A, G^{\star})> \xi(\Bar{\theta}_B, G^{\star})$. This shows that the function $\Bar{\theta} \mapsto  \xi(\Bar{\theta}, G^{\star})$ is decreasing.

The function $x \mapsto B( G_{\star},x)$ is increasing. Reasoning as we did above then shows that 
the function $\Bar{\theta} \mapsto  \xi(\Bar{\theta}, G_{\star})$ is increasing.
\end{proof}

By Proposition  \ref{prop2}, which 
grading rule maximizes  effort in a school depends on the average type of students attending this school: if the average type is 
greater than some cutoff  ($\Bar{\theta}^\dagger$), effort is maximized by the grading rule 
$G_{\star}$; if instead the average type is 
less  than this cutoff, effort is maximized by the grading rule 
$G^{\star}$. The logic is as follows. A student's incentives to exert effort are increasing in the difference 
 between $q_s^+(\hat{e}_s)$ 
 (employers' posterior belief that a student of $s$ who obtained the grade pass has productive value 1) 
 and $q_s^-(\hat{e}_s)$  (employers' posterior belief that a student of $s$ who obtained the grade fail has productive value 1). 
When the average type of students assigned to  school $s$ is high, then so 
is $q_s^+(\hat{e}_s)$; in this case, maximizing  $ q_s^+(\hat{e}_s) - q_s^-(\hat{e}_s)$ is  achieved by minimizing $q_s^-(\hat{e}_s)$, which in turn implies passing every student with productive value 1. Symmetrically, when 
the average type of students assigned to  school $s$ is low, so 
is $q_s^-(\hat{e}_s)$; maximizing  $ q_s^+(\hat{e}_s) - q_s^-(\hat{e}_s)$ is then achieved by maximizing $q_s^+(\hat{e}_s)$, which in turn implies failing  every student with productive value 0.

The final    proposition of this section  addresses  the following question. Suppose that  
two schools share the same  grading rule:
if the goal is to maximize total effort, 
should their  students be pooled or kept separated?

\begin{proposition} \label{prop3}
 For every  $G \in \mathbb{G}$, the function  
 $ \bar{\theta} \mapsto \xi(\bar{\theta}, G)$ is strictly concave.   
\end{proposition}

\begin{proof}
 Let $G \in \mathbb{G}$, 
$\bar{\theta}_A, \bar{\theta}_B \in (0, \alpha)$ with  $\bar{\theta}_A\neq  \bar{\theta}_B $, and $\lambda \in (0,1)$. The strict concavity of the function $x \mapsto B(G,x)$ together with the convexity of  $c'$
imply 
\begin{align*}
B\big(G, (1-\lambda) 
 (\bar{\theta}_A + \xi(\bar{\theta}_A, G))  +\lambda  (\bar{\theta}_B + \xi(\bar{\theta}_B, G))\big)&> (1-\lambda) 
B\big(G, \bar{\theta}_A + \xi(\bar{\theta}_A, G)\big) + \lambda  B\big(G, \bar{\theta}_B  + \xi(\bar{\theta}_B, G)\big) \\
&=    (1-\lambda)  c'\big( \xi(\bar{\theta}_A, G)  \big) +
 \lambda  c'\big( \xi(\bar{\theta}_B, G)  \big) 
  \\
&\geq  c'\big(  (1-\lambda) \xi(\bar{\theta}_A, G) + \lambda\xi(\bar{\theta}_B, G)  \big).
 \end{align*}
 Finally,
the function $e \mapsto B(G, (1-\lambda) \Bar{\theta}_A+  \lambda \Bar{\theta}_B  +e )- c'(e)$ is 
 strictly concave, positive  at $e=0$, and zero 
at   $e= \xi\big((1-\lambda) \Bar{\theta}_A+  \lambda \Bar{\theta}_B, G\big)$. So the highlighted 
 expression 
implies 
\[  (1-\lambda) \xi(\bar{\theta}_A, G) + \lambda\xi(\bar{\theta}_B, G)<  \xi\big((1-\lambda) \Bar{\theta}_A+  \lambda \Bar{\theta}_B, G\big).
\] 
\end{proof}

By Proposition  \ref{prop3}, if two  schools share the same grading rule, then pooling their students in a single school   increases  total  effort.
The basic intuition  is as follows. Separating students provides employers with information about students' types.
  This additional information   weakens the informational content of the grades, and,  consequently, also the benefit of exerting effort.

\section{Socially Optimal  School Systems}\label{sec:mainresult}

We saw in the previous section 
 that each school system induces a unique  effort profile: specifically,  
students in a  school $s$ exert effort $ \xi(\Bar{\theta}_s, G_s)$.
The school  system $\big(n, \phi, 
(G_{s})_{s=1}^n \big)$ therefore 
  generates social welfare
\[ 
W\big(n, \phi, 
(G_{s})_{s=1}^n \big):=
\sum_{s=1}^n m_{s} \big[ \Bar{\theta}_{s} + \xi(\Bar{\theta}_s, G_s) - c\big(  \xi(\Bar{\theta}_s, G_s)  \big)\big].
\]
A school system is \textit{socially optimal} if it maximizes social welfare among all possible school systems.
Our goal in this  section  is to characterize the socially optimal  school systems.

We will 
say that a school system  $\big(n, \phi, 
(G_{s})_{s=1}^n \big)$
has a \textit{single-tier structure} if  
$\bar{\theta}_{s_A}= \bar{\theta}_{s_B}$ for all $s_A,s_B \in \{1, \cdots, n\}$. 
A school system 
 $\big(n, \phi, 
(G_{s})_{s=1}^n \big)$
 is said 
to 
 have a 
 \textit{two-tier structure} 
 if there exist non-empty sets $S$ and 
$S'$  partitioning $\{1,\cdots, n\}$
such that   
\begin{itemize}
\item
$\bar{\theta}_{s_A}= \bar{\theta}_{s_B}< \bar{\theta}_{s'_A}= \bar{\theta}_{s'_B} $ for all $s_A, s_B \in S$ and 
$s'_A, s'_B \in S'$;
   \item 
   $\theta_j \leq \theta_{j'}$ whenever $\phi(j) \in S$ and  $\phi(j') \in S'$. 
\end{itemize} 
We then refer to 
$S$ as the bottom tier and to $S'$ as the top tier.
 In words,  a school system has a two-tier structure
  if every school belongs to one of 
 two tiers, such that: (i)
 all schools in a given tier   have  the same  average student type,  (ii) the
  type of every  student in a top-tier school is at least as a large as the type of any student in a bottom-tier school.
We can now state our first theorem.

\begin{theorem}\label{thm1}
Every socially optimal school system  has either a 
single-tier structure
  or a two-tier structure. In the latter case, schools in the top tier have grading rule $G_{\star}$, and schools in the bottom tier have grading rule $G^{\star}$. 
\end{theorem}

\begin{proof} 
Any school $s$ forming part of a given school system is such that $c'\big(  \xi(\Bar{\theta}_s, G_s)\big)=B\big(G_s, \Bar{\theta}_s+    \xi(\Bar{\theta}_s, G_s)\big)<1$. So   $ \xi(\Bar{\theta}_s, G_s)< 
    \overline{\overline{e}}
    $, where  
$ \overline{\overline{e}}$ denotes the unique effort such that $c'( \overline{\overline{e}})=1$.
Yet
$e\mapsto e-c(e)$ is increasing on $[0,  \overline{\overline{e}}]$. 
It follows that any school  
$s$ forming part of an efficient school system must be such that 
\[
G_s \in \arg \max_{G \in \mathbb{G}}~  \xi(\bar{\theta}_s, G).
\] Applying 
 Proposition \ref{prop1} then shows that 
any school  
$s$ forming part of an efficient school system is such that 
$G_s \in \{G_{\star}, G^{\star}\}$. Moreover, by Proposition \ref{prop2}, 
$G_s= G^{\star}$ if 
$\Bar{\theta}_s<\Bar{\theta}^\dagger$, while  
 $G_s= G_{\star}$ if 
$\Bar{\theta}_s>\Bar{\theta}^\dagger$.

Next, 
consider  $G \in \mathbb{G}$, 
 $\bar{\theta}_A, \bar{\theta}_B \in (0, \alpha)$ with  $\bar{\theta}_A\neq  \bar{\theta}_B $, and $\lambda \in (0,1)$. 
We then have
\begin{align*}
(1-\lambda) \big[ \xi(\bar{\theta}_A, G)   -  & c\big( \xi(\bar{\theta}_A, G)   \big)  \big] +  \lambda \big[ \xi(\bar{\theta}_B, G)   -  c\big( \xi(\bar{\theta}_B, G)   \big)  \big]  \\ &\leq 
(1-\lambda) \xi(\bar{\theta}_A, G) + \lambda \xi(\bar{\theta}_B, G) -  c\big(  (1-\lambda) \xi(\bar{\theta}_A, G) + \lambda \xi(\bar{\theta}_B, G) \big) \\
& < \xi\big((1-\lambda) \Bar{\theta}_A+  \lambda \Bar{\theta}_B, G\big) -c  \big(  \xi\big((1-\lambda) \Bar{\theta}_A+  \lambda \Bar{\theta}_B, G\big)\big),
\end{align*} where
the first inequality follows from the  concavity of the function 
$e\mapsto e-c(e)$, and  the second inequality
follows from 
 the fact that this function is increasing  on  $[0,  \overline{\overline{e}}]$
  while,  by Proposition \ref{prop3}, $ \xi\big((1-\lambda) \Bar{\theta}_A+  \lambda \Bar{\theta}_B, G\big)>(1-\lambda) \xi(\bar{\theta}_A, G) + \lambda\xi(\bar{\theta}_B, G)$.   
The inequalities above imply that if 
two schools with the same grading rule
form part of a socially optimal school system then their average  student types
have  to be equal.

Finally, we saw in the proof of Proposition \ref{prop2}
that  $\Bar{\theta} \mapsto  \xi(\Bar{\theta}, G^{\star})$ is decreasing, and $\Bar{\theta} \mapsto  \xi(\Bar{\theta}, G_{\star})$ increasing. Hence, if two 
schools $s_A$ and $s_B$ forming part of a given school system 
are such that 
(i) $G_{s_A}=G^{\star}$, (ii) $G_{s_B}=G_{\star}$, (iii) there exist a student of  $s_A$
and a student of $s_B$ such that the 
 type 
 of the former 
 is greater than the type of  the latter, then  swapping these students increases both  equilibrium effort in school $s_A$ ($e^*_{s_A}$) and equilibrium effort in school $s_B$ ($e^*_{s_B}$),  and therefore also 
$m_{s_A} [  e^*_{s_A} - c(e^*_{s_A})] + m_{s_B} [  e^*_{s_B} - c(e^*_{s_B})]$.
We conclude that 
if two 
schools $s_A$ and $s_B$ forming part of an socially optimal school system 
are such that  (i) and (ii) hold, then the type of each  student in $s_B$ 
must be at least as large as the type of any  student in $s_A$.
\end{proof}

Theorem \ref{thm1} combines insights from Propositions 
\ref{prop2} and \ref{prop3}.
 Suppose that either all students 
have a type  
which is larger  than the cutoff $\theta^\dagger$   in Proposition \ref{prop2}, or all students have a type  
which is smaller than this cutoff. Then, 
by the same proposition, every  school forming part of a socially optimal school system must adopt the same grading rule (either $G_{\star}$ or $G^{\star}$).
Yet, by Proposition  \ref{prop3}, if two  schools share the same grading rule, then pooling their students in a single school     increases  total effort. It follows that, in this case, every socially optimal school system must  essentially contain   a  single school. By contrast, if student types are very dispersed and low types are separated from high types then the efficient  grading rule might vary across  schools.  In this case, 
separating students in two groups may be  optimal. We next  illustrate Theorem \ref{thm1} in an example   with two types.

\paragraph{Example.} Suppose that every  student is  of one of two types, $\mu- \sigma$ and $\mu + \sigma$, with   $\mu$ representing the average student type. Then there exists 
$\sigma^\dagger$
such that  if $\sigma<\sigma^\dagger$  every  socially optimal school system has a 
single-tier structure, while  if $\sigma>\sigma^\dagger$  every  socially optimal school system has a 
two-tier structure.\footnote{The proof is straightforward, and combines  Theorem \ref{thm1} with the fact that $\Bar{\theta} \mapsto  \xi(\Bar{\theta}, G^{\star})$ is decreasing, and $\Bar{\theta} \mapsto  \xi(\Bar{\theta}, G_{\star})$ increasing. }

\section{Incentive Compatibility}\label{sec:IC}

This section explores incentive compatibility. 
Consider a designer (she)  choosing  the school system. As long as she  can observe students' types,   the designer can in principle implement the school system of her choice;
in particular, any  socially optimal school system can be implemented. On the other hand, if the designer is either unable to observe students' types, or cannot freely assign students to schools, then the designer is forced to let students  self-select into schools. 
Subsection \ref{subsec:ICwithout}  studies incentive compatibility  without transfers, while  Subsection \ref{subsec:ICwith}  studies incentive compatibility  with transfers.

\subsection{Incentive Compatibility without Transfers} \label{subsec:ICwithout}

To shorten notation, in this section let $e^*_s= \xi(\Bar{\theta}_s, G_s)$.
Given a school system
$\big(n, \phi, 
(G_{s})_{s=1}^n \big)$, let 
$U^*_s(\theta)$ denote a type-$\theta$ student's  maximum expected payoff from attending school $s$, that is, 
\[
U^*_s(\theta) :=U_s(\theta, e^*_s, e^*_s)=
[(\theta+e^*_s)g_s^1+ (1-\theta-e^*_s) g_s^0] q_s^+(e^*_s) + 
[1-(\theta+e^*_s)g_s^1- (1-\theta-e^*_s) g_s^0]  q_s^-(e^*_s)  -
 c(e^*_s).
 \]
A school system $\big(n, \phi, 
(G_{s})_{s=1}^n \big)$ is \textit{incentive compatible  without transfers}
if every student $j$ weakly prefers attending school $\phi(j)$ than any other school, that is, if  
for every $j\in \{1, \cdots, m\}$:
\[
U^*_{\phi(j)}(\theta_j) =\max_{s \in \{1, \cdots,n\} } ~ U^*_s(\theta_j).
\]

\begin{theorem} \label{thm2}
Any 
 school system $\big(n, \phi, 
(G_{s})_{s=1}^n \big)$  that is  incentive compatible without transfers
has a single-tier structure, and satisfies 
$e^*_{s_A}=e^*_{s_B}$
for all $s_A,s_B \in \{1, \cdots, n\}$.
\end{theorem}

\begin{proof} Consider a school system 
 $\big(n, \phi, 
(G_{s})_{s=1}^n \big)$. We have
\begin{align*}
 U^*_s(\theta) &=  q_s^-(e^*_s)+ \big[g_0^s + (\theta+e^*_s) (g_1^s- g_0^s) \big] \big[ q_s^+(e^*_s)- q_s^-(e^*_s)\big]- c(e^*_s)
\\ 
&= q_s^-(e^*_s)+ \big[g_0^s + (\theta+e^*_s) (g_1^s- g_0^s) \big] \left(\frac{c'(e^*_s)}{g_1^s- g_0^s} \right) - c(e^*_s)
\\ 
&= \zeta_s + \theta c'(e^*_s),
\end{align*}
where
\[
\zeta_s=  q_s^-(e^*_s)+ \big[g_0^s + e^*_s (g_1^s- g_0^s) \big] \left(\frac{c'(e^*_s)}{g_1^s- g_0^s} \right) - c(e^*_s).
\]
So if $s_A$ and $s_B$ are two schools forming part of the school system considered, then
\begin{equation}\label{eq:singlecrossing}
 U^*_{s_B}(\theta) - 
 U^*_{s_A}(\theta)= (\zeta_{s_B}- \zeta_{s_A}) + \theta\big(c'(e^*_{s_B})- c'(e^*_{s_A})  \big).
\end{equation}

Now reason by contradiction: suppose  that some school system  $\big(n, \phi, 
(G_{s})_{s=1}^n \big)$ is  incentive compatible without transfers, and we can find two  schools $s_A$ and $s_B$ forming part of this system such that 
$e^*_{s_B} >e^*_{s_A}$. Then,
 by \eqref{eq:singlecrossing}, every  student attending school $s_A$ 
must have a weakly lower type than all students attending school 
$s_B$, whence   
$\bar{\theta}_{s_A} \leq  \bar{\theta}_{s_B}$.
We thus obtain
\begin{align*}
    U^*_{s_B}(\bar{\theta}_{s_A}) = U_{s_B}(\bar{\theta}_{s_A}, &e^*_{s_B}, e^*_{s_B}) >  U_{s_B}(\bar{\theta}_{s_A}, e^*_{s_A}, e^*_{s_B})  \\
    &> U_{s_B}(\bar{\theta}_{s_A}, e^*_{s_A}, e^*_{s_A}) 
    \geq 
      \bar{\theta}_{s_A} + e^*_{s_A} -c( e^*_{s_A})
  =  
    U_{s_A}(\bar{\theta}_{s_A}, e^*_{s_A}, e^*_{s_A}) =U^*_{s_A}(\bar{\theta}_{s_A}).
\end{align*}
The first inequality captures the fact that effort $e^*_{s_B}$ is optimal in school $s_B$, while $e^*_{s_A}$ is not; the second inequality captures the benefit from increasing employers' conjecture regarding effort exerted in school $s_B$ (from $e^*_{s_A}$ to $e^*_{s_B}$); the third inequality captures the  benefit from increasing employers' belief concerning the average  student type in  a school (from $\bar{\theta}_{s_A}$ to $\bar{\theta}_{s_B}$).

The last highlighted expression combined with 
\eqref{eq:singlecrossing} gives 
$ U^*_{s_B}(\theta)> U^*_{s_A}(\theta)$ for all $\theta \geq \bar{\theta}_{s_A}$.
Now, some student $j$ attending school $s_A$ has type $\theta_j \geq \bar{\theta}_{s_A}$, and so $ U^*_{s_B}(\theta_j)> U^*_{\phi(j)}(\theta_j)$, contradicting our assumption that $\big(n, \phi, 
(G_{s})_{s=1}^n \big)$ is incentive compatible without transfers. 
This shows that if 
 $\big(n, \phi, 
(G_{s})_{s=1}^n \big)$ is   incentive compatible without transfers, then 
$e^*_{s_A}=e^*_{s_B}$
for all $s_A,s_B \in \{1, \cdots, n\}$. Furthermore,  
$\bar{\theta}_{s_A}= \bar{\theta}_{s_B}$, for the same reasoning as above shows that 
$\bar{\theta}_{s_B}> \bar{\theta}_{s_A}$
implies 
$ U^*_{s_B}(\theta)> U^*_{s_A}(\theta)$ for all $\theta \geq \bar{\theta}_{s_A}$.
\end{proof}

Theorem 
\ref{thm2} tells us  that if a designer is  unable to observe students' types and cannot use transfers  then she can do no better than placing  all students in one school.
The basic idea is simple: if the average type of students in  school A were greater than in school B, then all students would want to choose  school A so as to benefit from a more favorable prior in the eyes of employers.

\subsection{Incentive Compatibility with Transfers} \label{subsec:ICwith}

Transfers,  in the form of differentiated fees for example,  can play an important role in the design of school systems. We show in this subsection that, with transfers, Theorem \ref{thm1}  naturally extends,  in the sense that every \emph{constrained}  optimal school system has either a 
single-tier structure
  or a two-tier structure.

  To streamline  the exposition of our next result, we  examine in this subsection the continuous version of our model.
There is  a continuum of students, with types distributed according to the atomless 
distribution $F$ 
such that   $\text{supp}\; F= [0, \alpha]$. The assignment of students to schools is
 captured by a pair $(n, \phi)$, where 
 $n \in \mathbb{N}$ 
and  
  $\phi$
 is a measurable  mapping from  
  $(0, \alpha)$ to  $\{1, \cdots, n\}$
  such that $m_s:=\int_{\phi^{-1}(s)} dF(\theta)>0$ for $s=1, \cdots, n$: we interpret $\phi$ as saying   that   type-$\theta$ students  attend school $\phi(\theta)$. 
The rest is as in the discrete model of Section \ref{sectionmodel}.

A school system $\big(n, \phi, 
(G_{s})_{s=1}^n \big)$ is \textit{incentive compatible  with transfers}
if
there exist school  fees $t_1, \cdots, t_n$ such that, for every $\theta \in (0, \alpha)$,
  type-$\theta$ students weakly prefer attending school $\phi(\theta)$ than any other school, that is,
\[
U^*_{\phi(\theta)}(\theta)  - t_{\phi(\theta)}
=\max_{s \in \{1, \cdots,n\} } ~ U^*_s(\theta)-t_s, \qquad \forall\; \theta \in (0, \alpha).
\]

We say that a school system is \textit{constrained  optimal} if it maximizes social welfare within the class of   school systems that are incentive compatible  with transfers.

\begin{theorem}\label{thm3}
Every constrained-optimal school system  has either a 
single-tier structure
  or a two-tier structure.\footnote{In the continuous version of the model, a school system has a two-tier structure
  if every school belongs to one of 
 two tiers, such that: (i)
 all schools in a given tier   have  the same  average student type,  (ii) the
  type of almost every   student in a top-tier school is at least as a large as the type of almost any  student in a bottom-tier school.
In other words, sets of students that have  measure zero play no role. } In the latter case, schools in the top tier have grading rule $G_{\star}$, and schools in the bottom tier have grading rule $G^{\star}$. 
\end{theorem}

 The proof of Theorem \ref{thm3}  is long and therefore relegated to the appendix.
We first show that 
if
  a school system is  socially optimal  and has either  a single-tier structure, 
  or a two-tier structure such that  
  equilibrium effort in the top-tier schools is at least as large as  equilibrium effort in the bottom-tier schools, then  
 this school system is  also constrained  optimal.
However, if
  a school system is  socially optimal  and 
 has
 a two-tier structure such that 
  equilibrium effort in the top-tier schools is less than  equilibrium effort in the bottom-tier schools, then this  school system is not  incentive compatible with transfers.
 The proof of Theorem \ref{thm3} shows  that   in this  case a constrained-optimal school system is retrieved   by  
moving intermediate-type students from the top-tier schools to the bottom-tier schools. In other words, 
 the incentive compatibility constraints
 distort the 
 socially optimal school system by  
  reducing the size of the top tier and increasing  that of the bottom tier.


\section{An Example with Three  Productive Values}\label{sec:morethantwo}

The grading rule $G_{\star}$ passes  
every 
student with productive value $1$ as well as a fraction $1-b$ of students with productive value $0$ and can thus be interpreted as a ``lenient" grading rule. 
By contrast, the grading rules $G^{\star}$  fails every student  
with productive value $0$ as well as a fraction $1-b$  of students with productive value $1$ and may thus   be interpreted as a ``tough" grading rule. 
We saw in Proposition \ref{prop2} that whether the tough grading rule $G^{\star}$  is preferable to the lenient grading rule  $G_{\star}$ depends on the type of students in the school considered:
when  the average student type is high,   $G_{\star}$ induces more effort than 
 $G^{\star}$, and vice versa when 
  the average student type is low.  We now show through an example that the logic underlying Proposition \ref{prop2} extends beyond two  productive 
values.

Consider a modified version of the model of Section \ref{sectionmodel} with three productive values, $v_1<v_2<v_3$. Let $(\underline{\pi}_1, \underline{\pi}_2, \underline{\pi}_3)$ and $(\overline{\pi}_1, \overline{\pi}_2, \overline{\pi}_3)$ represent 
two probability distributions over $\{v_1, v_2, v_3\}$, such that $(\overline{\pi}_1, \overline{\pi}_2, \overline{\pi}_3)$ first-order stochastically dominates $(\underline{\pi}_1, \underline{\pi}_2, \underline{\pi}_3)$.
At the end of school, each student's productive
value is either $v_1$, $v_2$, or $v_3$. We assume that the distribution of student $j$'s productive
value  is a convex combination of  $\underline{\pi}$
and $\overline{\pi}$:
\[
(1-\theta_j- e_j) \underline{\pi}+ (\theta_j+ e_j) \overline{\pi},
\]
where $\theta_j \in (0, \alpha)$ denotes student $j$'s type, and $e_j \in [0, 1-\alpha]$ his effort. Thus, effort increases a student's chances of acquiring  a high productive value.

As in Section \ref{sectionmodel}, let $b$ denote a parameter in $(0,1)$, and  
  $G_{\star}$ and $G^{\star}$ be two  grading rules, defined as follows. Under both grading rules, 
a student with productive value $v_1$ obtains the grade $C$ with probability $1$, a student with productive value 
$v_2$ obtains the grade $B$ with probability $b$, 
and  a student  with productive value $v_3$ obtains the grade $A$ with probability $1$.
However, whereas under 
$G_{\star}$ a  student with productive value 
$v_2$ obtains the grade $A$ with probability $1-b$, under  
$G^{\star}$ a  student with productive value 
$v_2$ obtains the grade $C$ with probability $1-b$.
In words,   $G_{\star}$ occasionally inflates $B$ to $A$, whereas   $G^{\star}$ occasionally deflates  $B$ to $C$.
In line with 
Section \ref{sectionmodel},  the grading rules $G^{\star}$ and $G_{\star}$  may thus be interpreted as tough and lenient, respectively.
The following result mirrors 
Proposition \ref{prop2}:

\begin{proposition}\label{propappendix}
    There exists $\Bar{\theta}^\dagger$ such that  $\xi(\Bar{\theta}, G^{\star})> \xi(\Bar{\theta}, G_{\star})$ for $\Bar{\theta}<\Bar{\theta}^\dagger$,  and 
$\xi(\Bar{\theta}, G_{\star})> \xi(\Bar{\theta}, G^{\star})$ for $\Bar{\theta}>\Bar{\theta}^\dagger$.
\end{proposition}

\begin{proof} For any school $s$ forming part of a given school system,  let 
 $q_s^A(\hat{e}_s)$ denote employers' posterior expectation of the productive value of  a student of $s$ who obtained the grade $A$ (given  employers' conjecture that students of $s$ exert effort $\hat{e}_s$). Let $q_s^B(\hat{e}_s)$ and $q_s^C(\hat{e}_s)$ be similarly defined for grades $B$ and $C$, respectively.
 If $G_s=G^{\star}$,
a student's marginal benefit of exerting effort 
 in school $s$ is given by
 \[
 B_s( G^{\star}, \hat{e}_s):=  (\overline{\pi}_1- \underline{\pi}_1) q_s^C(\hat{e}_s) + (\overline{\pi}_2- \underline{\pi}_2) \big[b q_s^B(\hat{e}_s) + (1-b) q_s^C(\hat{e}_s) \big] +  (\overline{\pi}_3- \underline{\pi}_3)  q_s^A(\hat{e}_s);
 \]
 if $G_s=G_{\star}$,
the  marginal benefit of  effort 
 is given by
 \[
 B_s( G_{\star}, \hat{e}_s):=  (\overline{\pi}_1- \underline{\pi}_1) q_s^C(\hat{e}_s) + (\overline{\pi}_2- \underline{\pi}_2) \big[b q_s^B(\hat{e}_s) + (1-b) q_s^A(\hat{e}_s) \big] +  (\overline{\pi}_3- \underline{\pi}_3)  q_s^A(\hat{e}_s).
 \]
For any school $s$ forming part of a given school system,  $e$ is an equilibrium effort in school $s$ if and only if 
\[
c'(e)=  B_s( G_s, \hat{e}_s).
\]
The function $e \mapsto B_s(G_s, + e )- c'(e)$ is 
 strictly concave, positive  at $e=0$, and negative  at $e=1-\alpha$. So  the previous equation  has a unique solution.
 This solution is   
 determined by $\bar{\theta}_s$ and $G_s$, and so can be denoted by $\xi(\Bar{\theta}_s, G_s)$.

Now
consider a school $s$ such that $G_s=G^{\star}$. Then  $q_s^A(\hat{e}_s)=v_3$, $q_s^B(\hat{e}_s)=v_2$, and 
$q_s^C(\hat{e}_s)$ is increasing in $\hat{e}_s$. As
$\overline{\pi}_1- \underline{\pi}_1<0$ and either $\overline{\pi}_2- \underline{\pi}_2<0$ or $(1-b)|\overline{\pi}_2- \underline{\pi}_2|<|\overline{\pi}_1- \underline{\pi}_1|$, it follows that  
the function  $e \mapsto B_s( G^{\star},  e )$ is decreasing. This, in turn, 
implies that  the function $\Bar{\theta} \mapsto  \xi(\Bar{\theta}, G^{\star})$ is decreasing.

Next, 
consider a school $s$ such that $G_s=G_{\star}$. Then  $q_s^A(\hat{e}_s)$ is  increasing in $\hat{e}_s$, $q_s^B(\hat{e}_s)=v_2$, and 
$q_s^C(\hat{e}_s)=v_1$. As
$\overline{\pi}_3- \underline{\pi}_3>0$ and either 
$\overline{\pi}_2- \underline{\pi}_2>0$ or 
$(1-b)|\overline{\pi}_2- \underline{\pi}_2|<|\overline{\pi}_3- \underline{\pi}_3|$, it follows that  
the function  $e \mapsto B_s( G_{\star},  e )$ is increasing. This, in turn, 
implies that  the function $\Bar{\theta} \mapsto  \xi(\Bar{\theta}, G_{\star})$ is increasing.  Combining the previous remarks concludes the proof.
\end{proof}

While Proposition 
\ref{propappendix} extends Proposition  \ref{prop2} to an environment with three productive values, moving beyond binary productive values introduces various   challenges. Firstly, 
there appears to be no natural  notion of 
leniency  ordering the set of all  possible grading rules.
Secondly, beyond binary  productive values, 
leniency alone does  not to capture all relevant aspects of a grading rule. To illustrate the latter point, consider the grading rules
  $\Tilde{G}_{\star}$ and $\Tilde{G}^{\star}$ defined as follows. Under both grading rules, 
a student with productive value $v_2$ obtains the grade $B$ with probability $1$.
Under 
  $\Tilde{G}_{\star}$, 
a  student with productive value 
$v_3$ obtains the grade $A$ with probability $1$, while   
 a  student with productive value 
$v_1$ obtains the grade $C$ with probability $b$ and the grade $B$ with probability $1-b$. By contrast, under  $\Tilde{G}^{\star}$, 
a  student with productive value 
$v_1$ obtains the grade $C$ with probability $1$, while   
 a  student with productive value 
$v_3$ obtains the grade $A$ with probability $b$ and the grade $B$ with probability $1-b$. 
In words,   $\Tilde{G}_{\star}$ occasionally inflates $C$ to $B$, whereas   $\Tilde{G}^{\star}$ occasionally deflates  $A$ to $B$.
The grading rule  $\Tilde{G}_{\star}$ may therefore, in a sense, be viewed as  more lenient than  $\Tilde{G}^{\star}$. Yet, whether  
$\Tilde{G}_{\star}$ induces more effort than $\Tilde{G}^{\star}$ or the way around need not be monotonic in a school's average student type.
To see why, consider  a school $s$ with $G_s=\Tilde{G}_{\star}$. Increasing $\bar{\theta}_s$ increases $q_s^B(\hat{e}_s)$, which raises the gain $q_s^B(\hat{e}_s)-q_s^C(\hat{e}_s)$ from obtaining grade $B$ instead of $C$, but reduces the gain  $q_s^A(\hat{e}_s)-q_s^B(\hat{e}_s)$  from obtaining grade $A$ instead of $B$. Which of these effects dominates the other depends on the parameters.

\section{Conclusion}\label{sec:conc}

This paper proposes  a framework  to assess how to optimally sort and grade 
students. The premise of our model is that an individual's productive value results from a combination of ability and effort. But potential employers observe neither productive value, nor ability, nor  effort:
they have to make decisions based solely on the school which an individual attended and the grades that he obtained. The way in which students are sorted and graded therefore   determines  incentives to exert effort.

We show that if a planner can freely assign students to schools (as in, e.g., Austria, England, Germany and  Switzerland) then any  socially optimal school system either randomly assigns students to schools or separates students in just two ability groups. 
If the variance in student ability is small, then 
a random assignment is optimal; otherwise students whose ability is above a cutoff should be separated from those whose ability is below the cutoff. In the latter case, 
bottom-tier students are evaluated through
a tough grading rule, and  top-tier students  through a lenient grading rule.

We then  study the case in which students can choose the school they go to (as in, e.g., France and Italy). We show that 
in the absence of transfers, a planner can do no better than randomly assigning students to schools. However, with transfers, any constrained-optimal school system has the same qualitative
features as when the
planner 
can 
freely assign students to schools.

\newpage

\section*{Appendix: Proof of Theorem
\ref{thm3}}

We start with a couple of  definitions.  We will say that 
a  school system  $\big(n, \phi, 
(G_{s})_{s=1}^n \big)$ is \textit{$r$-regular} if there exist  $r$ non-empty sets $S_1, \cdots, S_r$  partitioning $\{1,\cdots, n\}$,   cutoffs $0= \hat{\theta}_0<\hat{\theta}_1< \cdots<\hat{\theta}_{r-1}< \hat{\theta}_r=\alpha$, and effort levels $e_1 \leq \cdots \leq e_r$
such that 
\begin{itemize}
    \item   for every  $k\in \{1, \cdots, r\}$, 
    $\phi(\theta) \in S_k$ for  all $\theta \in  ( \hat{\theta}_{k-1}, \hat{\theta}_k)$;
 \item  $\phi (\hat{\theta}_k) \in S_k \cup S_{k+1}$ for all  $k\in \{0, \cdots, r-1\}$, and  $\phi (\hat{\theta}_r) \in S_r$ ;
    \item  
$e^*_{s}= e_k $ for all school $s \in S_k$, $k\in \{1, \cdots, r\}$.
\end{itemize}
A school system is  \textit{regular} if it is   $r$-regular for some $r$. 
We then refer to $\big( (\hat{\theta}_k)_{k=0}^r, (e_k)_{k=1}^r \big)$  as the \textit{signature} of the regular school system, and define 
\[
M_k:= \int_{\phi^{-1}(S_k)} dF(\theta)= \int_{\hat{\theta}_{k-1}}^{\hat{\theta}_{k}} dF(\theta),  \quad ~~ k \in  \{1, \cdots, r\}.
\]
Notice that
the social welfare which 
a regular school system  generates is entirely determined by its signature: 
\begin{equation*}
W\big(n, \phi, 
(G_{s})_{s=1}^n \big):=
\sum_{k=1}^r  M_k
\left[ \frac{ \int_{\phi^{-1}(S_k)} \theta dF(\theta)}{M_k} +e_k - c( e_k)  \right].
\end{equation*}

Finally, given two regular school systems $\big(n, \phi, 
(G_{s})_{s=1}^n \big)$ and $\big(n', \phi', 
(G'_{s})_{s=1}^{n'} \big)$ with 
signatures  $\big( (\hat{\theta}_k)_{k=0}^r, (e_k)_{k=1}^r \big)$ and  $\big( (\hat{\theta}'_k)_{k=0}^{r'}, (e'_k)_{k=1}^{r'} \big)$, 
we will say that  $\big(n', \phi', 
(G'_{s})_{s=1}^{n'} \big)$  is 
 a  \textit{non-decreasing  contraction} of  $\big(n, \phi, 
(G_{s})_{s=1}^n \big)$ if 
\begin{itemize}
    \item $\{  \hat{\theta}'_k   \}_{k=1}^{r'} \subseteq \{  \hat{\theta}_k   \}_{k=1}^{r}$; 
    \item   $e'_k \geq   \frac{\sum_{l=k^{-}}^{k^+} M_le_l
}{\sum_{l=k^{-}}^{k^+} M_l}$, for all $k \in \{1, \cdots, r'\}$, where $k^-$ and  $k^+$ are the indices such that  $\hat{\theta}'_{k-1}= \hat{\theta}_{k^-}$ and  $\hat{\theta}'_{k}= \hat{\theta}_{k^+}$.
\end{itemize}
By extension, we say that $\big(n', \phi', 
(G'_{s})_{s=1}^{n'} \big)$  is 
 an  \textit{increasing contraction} of  $\big(n, \phi, 
(G_{s})_{s=1}^n \big)$ if it is a 
non-decreasing contraction  of  the latter school system and  $e'_k >  \frac{\sum_{l=k^{-}}^{k^+} M_le_l
}{\sum_{l=k^{-}}^{k^+} M_l}$  for at least one $k \in \{1, \cdots, r'\}$.

The next two lemmas  shed light on the above definitions.

\begin{lemma}\label{lemma:regss1}
    A school system is incentive compatible with transfers if and only if it is regular.
\end{lemma}

\begin{proof}
We established in the proof of Theorem \ref{thm2} that, given a school system  $\big(n, \phi, 
(G_{s})_{s=1}^n \big)$, 
\[
 U^*_s(\theta) = \zeta_s + \theta c'(e^*_s),
 \]
for some $\zeta_s$ independent of $\theta$. As $c'$ is increasing, 
the lemma
 follows from standard mechanism design arguments.
\end{proof}

\begin{lemma} \label{lemma:regss2}
    If   $\big(n', \phi', 
(G'_{s})_{s=1}^{n'} \big)$  is 
 a  non-decreasing contraction   of  $\big(n, \phi, 
(G_{s})_{s=1}^n \big)$ then 
 \newline $W\big(n', \phi', 
(G'_{s})_{s=1}^{n'} \big)\geq W\big(n, \phi, 
(G_{s})_{s=1}^n \big)$.
 If   $\big(n', \phi', 
(G'_{s})_{s=1}^{n'} \big)$  is 
 an increasing contraction   of 
 $\big(n, \phi, 
(G_{s})_{s=1}^n \big)$ then $W\big(n', \phi', 
(G'_{s})_{s=1}^{n'} \big)> W\big(n, \phi, 
(G_{s})_{s=1}^n \big)$.
\end{lemma}

\begin{proof}
    We saw in the proof of Theorem \ref{thm1} that there exists  $\overline{\overline{e}}$ such that 
any school $s$ forming part of a given school system is such that    $e^*_s< 
    \overline{\overline{e}}
    $, and
$e\mapsto e-c(e)$ is increasing and concave  on $[0,  \overline{\overline{e}}]$.
\end{proof}

\quad

\begin{proof}[Proof of Theorem \ref{thm3}]
Pick a regular school system  
with signature  $\big( (\hat{\theta}_k^{(0)})_{k=0}^r, (e^{(0)}_k)_{k=1}^r \big)$. Suppose 
that  $\Sigma^{(0)}$  does not have a single-tier structure, nor a two-tier structure.
We will prove that $\Sigma^{(0)}$ is not constrained optimal.
In line with previous notation, let 
$S^{(0)}_k$ denote   the set of schools 
such that   $\phi(\theta) \in S^{(0)}_k$ for  all $\theta \in  ( \hat{\theta}^{(0)}_{k-1}, \hat{\theta}^{(0)}_{k})$.
As  
$\Sigma^{(0)}$ has  neither a  single nor  a two-tier structure, one of the following cases holds:
\begin{enumerate}
    \item[Case 1.] $\Sigma^{(0)}$ is $r$-regular with $r\geq 3$; 
    \item  [Case 2.] $\Sigma^{(0)}$ is $2$-regular, and   one of  $S_1^{(0)}$  and  $S_2^{(0)}$ contains two schools whose average student types differ;
\item [Case 3.]
we can find two schools forming part of 
$\Sigma^{(0)}$, $s_A$ and $s_B$ say, as well as two groups of students with positive mass, $A$ and $B$ say, such that (i) students in group $A$ attend school $s_A$, (ii) students in group $B$ attend school $s_B$, (iii) 
$\bar{\theta}_{s_A}<\bar{\theta}_{s_B}$, (iv) the type of any    students in  group $A$  is   larger than the  type of every  student in  group $B$.  
\end{enumerate}

We treat Case 1 below; the other cases are similar.
To streamline the exposition, suppose that 
$\Sigma^{(0)}$
is  $3$-regular,  with signature 
$\big((0,   \hat{\theta}^{(0)}_1, \hat{\theta}^{(0)}_2,  \alpha),   (e^{(0)}_1, e^{(0)}_2, e^{(0)}_3 )\big)$, where $ \hat{\theta}^{(0)}_1<\bar{\theta}^\dagger <  \hat{\theta}^{(0)}_2$ (with $\bar{\theta}^\dagger$ defined in Proposition \ref{prop2}).\footnote{The case in which $\bar{\theta}^\dagger \in \{\hat{\theta}^{(0)}_1, \hat{\theta}^{(0)}_2 \}$ is  
simpler; we focus here on the more complicated case.}

\quad 

\noindent \underline{ \textit{Step 1:}} 
Let $\Sigma^{(1)}$ denote the school system obtained from   $\Sigma^{(0)}$ by merging  all schools in the set $S^{(0)}_3$ into one school with grading rule $G_{\star}$. As  $\hat{\theta}^{(0)}_2> \bar{\theta}^\dagger$, all students in $S^{(0)}_3$ have   type greater than   $\bar{\theta}^\dagger$, whence, by Propositions 
\ref{prop1}-\ref{prop3}: $e^{(1)}_3\geq e^{(0)}_3$.
It follows that $\Sigma^{(1)}$ is a non-decreasing contraction of $\Sigma^{(0)}$.

\quad 

\noindent \underline{ \textit{Step 2:}} 
Let $S^{(1)+}_2$ denote the subset of   $S^{(1)}_2$ comprising the schools whose average student type is at least as large as 
$\bar{\theta}^\dagger$, and $S^{(1)-}_2$  the  subset of   $S^{(1)}_2$ comprising the schools whose average student type is less than 
$\bar{\theta}^\dagger$.
Now
let $\Sigma^{(2)}$ denote the school system obtained from   $\Sigma^{(1)}$ by   (i) merging  all schools in the set $S^{(1)-}_2$  into one school,  $s_A$ say,  with grading rule $G^{\star}$, and (ii) merging  all schools in the set $S^{(1)+}_2$  into one school,  $s_B$ say,  with grading rule $G_{\star}$. By Propositions 
\ref{prop1}-\ref{prop3}: 
\begin{equation}\label{eqtrainnevy1}
    \min \{e^*_{s_A}, e^*_{s_B} \}\geq e^{(1)}_2. 
\end{equation}

\quad 

\noindent \underline{ \textit{Step 3:}} 
Let $\Sigma^{(3)}$ denote the school system obtained from   $\Sigma^{(2)}$ by swapping  students between $s_A$ and $s_B$
so as to create two schools, $s_A'$ and $s_B'$ say, such that (i) $m_{s_A'}=m_{s_A}$, (ii)  $m_{s_B'}=m_{s_B}$, and (iii) the type of each student in $s_B'$ is greater than the type of any student
in $s_A'$. We saw in the proof of Proposition \ref{prop2}
that  $\Bar{\theta} \mapsto  \xi(\Bar{\theta}, G^{\star})$ is decreasing, and $\Bar{\theta} \mapsto  \xi(\Bar{\theta}, G_{\star})$ increasing. Therefore: $e^*_{s_A'} \geq e^*_{s_A}$ and  $e^*_{s_B'} \geq e^*_{s_B}$. Then  \eqref{eqtrainnevy1} yields
\begin{equation}\label{eqtrainnevy2}
    \min \{e^*_{s_A'}, e^*_{s_B'} \}\geq e^{(1)}_2. 
\end{equation}

\quad 

\noindent \underline{ \textit{Step 4:}} 
Let $\Sigma^{(4)}$ denote the school system obtained from   $\Sigma^{(3)}$ as follows. If $e^*_{s_B'}<e^*_{s_A'}$, replace the grading rule of 
 $s_A'$ 
by the grading rule $(0,\Tilde{b})$, with  $\Tilde{b}<b$ chosen so as to 
 ensure  $e^*_{s_B'}=e^*_{s_A'}$.  If instead $e^*_{s_A'}<e^*_{s_B'}$, replace the grading rule of 
 $s_B'$ 
by the grading rule $(1-\Tilde{b},1)$, with  $\Tilde{b}<b$ chosen so as to 
 ensure  $e^*_{s_B'}=e^*_{s_A'}$. Then,  by \eqref{eqtrainnevy2}, $\Sigma^{(4)}$ is a non-decreasing contraction of $\Sigma^{(1)}$.

\quad 

\noindent \underline{ \textit{Step 5:}} 
Let $S^{(4)+}$ denote the subset of   schools in $\Sigma^{(4)}$
 whose average student type is at least as large as 
$\bar{\theta}^\dagger$, and  $S^{(4)-}$  
the subset of   schools in $\Sigma^{(4)}$
 whose average student type is less than 
$\bar{\theta}^\dagger$. Now
let $\Sigma^{(5)}$ denote the school system obtained from   $\Sigma^{(4)}$
by   (i) merging  all schools in  $S^{(4)-}$  into one school,  $s''_A$ say,  with grading rule $G^{\star}$, and (ii) merging  all schools in $S^{(4)+}$ into one school,  $s_B''$ say,  with grading rule $G_{\star}$. By Propositions 
\ref{prop1}-\ref{prop3}: 
\begin{equation}\label{eqmonbard}
   e^*_{s_A''}\geq   \frac{\sum_{s \in S^{(4)-}} m_s e^*_s}{ \sum_{s \in S^{(4)-}} m_s}
  ,  \quad 
  e^*_{s_B''}\geq   \frac{\sum_{s \in S^{(4)+}} m_s e^*_s}{ \sum_{s \in S^{(4)+}} m_s},
\end{equation}
and one of those inequalities must be  strict.

\quad 

\noindent \underline{ \textit{Step 6:}}  
Let $\Sigma^{(6)}$ denote the school system obtained from   $\Sigma^{(5)}$ as follows. If 
$ e^*_{s_B''}\geq   e^*_{s_A''}$, keep  $\Sigma^{(5)}$  unchanged. If $ e^*_{s_B''}<   e^*_{s_A''}$,  increase the cutoff type separating students between the two schools so as to ensure  $ e^*_{s_B''} =e^*_{s_A''}$ (recall, 
 $\Bar{\theta} \mapsto  \xi(\Bar{\theta}, G^{\star})$ is decreasing, and $\Bar{\theta} \mapsto  \xi(\Bar{\theta}, G_{\star})$ increasing).
Then, by  \eqref{eqmonbard},   
 $\Sigma^{(6)}$ is an increasing contraction of 
$\Sigma^{(4)}$.

\quad 

As  $\Sigma^{(1)}$ is a non-decreasing contraction of $\Sigma^{(0)}$, 
 $\Sigma^{(4)}$ is a non-decreasing contraction of $\Sigma^{(1)}$, and 
 $\Sigma^{(6)}$ is an increasing contraction of 
$\Sigma^{(4)}$, we conclude that  $\Sigma^{(6)}$ is an increasing contraction of 
$\Sigma^{(0)}$. By Lemma \ref{lemma:regss2}, the social welfare generated by   $\Sigma^{(6)}$ is thus strictly larger than the social welfare generated by   $\Sigma^{(0)}$.
Yet, by  Lemma \ref{lemma:regss1}, $\Sigma^{(6)}$ is incentive compatible.  This shows that $\Sigma^{(0)}$ is not constrained optimal.
\end{proof}

\newpage


\bibliographystyle{apecon}
\bibliography{library}

\end{document}